\documentclass[runningheads,a4paper]{llncs}

\usepackage{amssymb}
\usepackage{graphicx}
\usepackage{comment}
\usepackage{enumitem}
\usepackage{lipsum}
\usepackage{ntheorem}
\usepackage{caption}
\usepackage{lipsum}

\usepackage{hyperref}
\newcommand{\orcid}[1]{\href{https://orcid.org/#1}{\includegraphics[width=8pt]{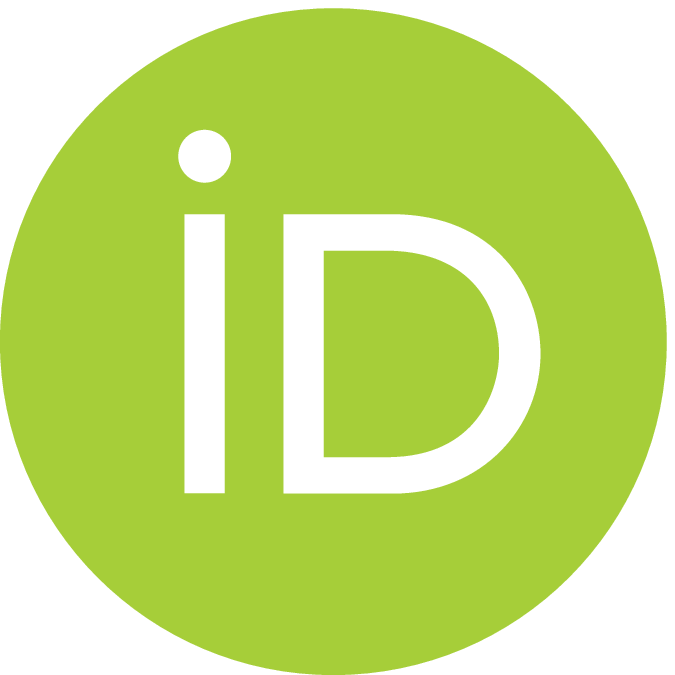}}}

\usepackage{url}
\urldef{\mailsa}\path|shenmaier@mail.ru|

\theorembodyfont{\upshape}
\newtheorem{rem}{Remark.}
\theorembodyfont{\itshape}
\newtheorem*{definition*}{Definition.}
\newtheorem{fact}{Fact.}

\AtBeginDocument{%
  \setlength{\oddsidemargin}{\dimexpr(\paperwidth-\textwidth)/2-1in}%
  \setlength{\evensidemargin}{\oddsidemargin}%
  \setlength{\topmargin}{%
    \dimexpr(\paperheight-\textheight)/2-\headheight-\headsep-1in}%
}

\renewenvironment{abstract}{
  \small
  \list{}{
    \setlength{\leftmargin}{.5cm}%
    \setlength{\rightmargin}{\leftmargin}%
  }%
  \item\relax
}

\begin{document}

\mainmatter  % start of an individual contribution

% first the title is needed
\title{Efficient PTAS for the Maximum Traveling Salesman Problem in a Metric Space of Fixed Doubling Dimension}

% a short form should be given in case it is too long for the running head
\titlerunning{Efficient PTAS for Max TSP in a Doubling Space}

% the name(s) of the author(s) follow(s) next
%
% NB: Chinese authors should write their first names(s) in front of
% their surnames. This ensures that the names appear correctly in
% the running heads and the author index.
%
\author{Vladimir Shenmaier \orcid{0000-0002-4692-1994}}
% ORCID of the author: http://orcid.org/0000-0002-4692-1994
%
\authorrunning{V.\,V. Shenmaier}
% (feature abused for this document to repeat the title also on left hand pages)

% the affiliations are given next; don't give your e-mail address
% unless you accept that it will be published
%\institute{Sobolev Institute of Mathematics, 4 Koptyug avenue, 630090 Novosibirsk, Russia\\
\institute{Sobolev Institute of Mathematics, Novosibirsk, Russia\\
\mailsa}

%
% NB: a more complex sample for affiliations and the mapping to the
% corresponding authors can be found in the file "llncs.dem"
% (search for the string "\mainmatter" where a contribution starts).
% "llncs.dem" accompanies the document class "llncs.cls".
%

\tocauthor{Vladimir Shenmaier}
\maketitle

\begin{abstract}{\bf Abstract.}
The maximum traveling salesman problem (Max~TSP) consists of finding a Hamiltonian cycle with the maximum total weight of the edges in a given complete weighted graph.
This problem is APX-hard in the general metric case but admits polynomial-time approximation schemes in the geometric setting, when the edge weights are induced by a vector norm in fixed-dimensional real space.
We propose the first approximation scheme for Max~TSP in an arbitrary metric space of fixed doubling dimension.
The proposed algorithm implements an efficient PTAS which, for any fixed $\varepsilon\in(0,1)$, computes a $(1-\varepsilon)$-approximate solution of the problem in cubic time.
Additionally, we suggest a cubic-time algorithm which finds asymptotically optimal solutions of the metric Max~TSP in fixed and sublogarithmic doubling dimensions.
\end{abstract}

\section{Introduction}
The maximum traveling salesman problem can be formulated as follows:\medskip

\noindent\textbf{Max~TSP.}
Given an $n$-vertex complete weighted (directed or undirected) graph $G$ with non-negative edge weights, find a Hamiltonian cycle in $G$ with the maximum total weight of the edges.\medskip

Max~TSP is the maximization version of the classic traveling salesman problem (TSP) and, like TSP, is among the most intensively researched NP-hard problems in computer science.
In this paper, we consider the metric Max~TSP, i.e., the special case in which the edge weights satisfy the triangle inequality and the symmetry axiom.\medskip

\noindent\textbf{Related work.}
Max~TSP has been actively studied since the 1970s.
%The first approximation algorithms for this problem was proposed by Fisher, Nemhauser, and Wolsey \cite{Fisher}, Kovalev and Kotov \cite{KK1981}, Serdyukov \cite{Serdyukov1984,Serdyukov1985,Serdyukov1987}, and Kostochka \cite{Serdyukov1985}.
The approximation factors of currently best polynomial-time algorithms in different cases are: $2/3$ for arbitrary asymmetric weights \cite{Kaplan}; $7/9$ for arbitrary symmetric weights \cite{Paluch79}; $35/44$ for the asymmetric metric case \cite{Kowalik3544}; and $7/8$ for the metric case~\cite{Kowalik78}.

On the complexity side, Max~TSP is APX-hard even in a metric space with distances $1$ and $2$: It follows from the corresponding result for TSP \cite{Papadimitriou1993,EK}.
The problem remains NP-hard in the geometric setting when the vertices of the input graph are some points in space $\mathbb R^3$ and the distances between them are induced by Euclidean norm~\cite{Fekete}.
The proof of this fact implies that the Euclidean Max~TSP does not admit a fully polynomial-time approximation scheme (FPTAS) unless P$=$NP. 

However, there exists a polynomial-time algorithm which computes asymptotically optimal solutions of the Euclidean problem in any fixed dimension~\cite{Serdyukov1987}.
The relative error of this algorithm is estimated as $c_d/n^{\frac{2}{d+1}}$, where $d$ is the dimension of space and $c_d$ is some constant depending on $d$.
In~\cite{Shen2010,Shen2014}, this result is extended to the case when the edge weights are induced by any (unknown) vector norm.
It follows that Max~TSP in a fixed-dimensional normed space admits an efficient polynomial-time approximation scheme (EPTAS).
Ano\-ther approach to constructing close-to-optimal solutions of the geometric Max~TSP is based on the algorithmic properties of this problem in a polyhedral space \cite{Serdyukov1997,Barvinok}.
This approach leads to a scheme EPTAS for the case of a ``fixed norm'', when it is possible to approximate the distances between vertices by a polyhedral metric.

Note that the usual traveling salesman problem is not in APX \cite{Sahni}.
The Euclidean TSP is NP-hard already in $\mathbb R^2$ \cite{Papadimitriou1977} but admits approximation schemes PTAS for each fixed dimension \cite{Arora,Mitchell}.
%The approximation ratio of the best deterministic algorithm for the metric TSP is $3/2$ \cite{Christofides,Serdyukov1978} but, in the case of fixed doubling dimensions, this problem admits a scheme PTAS \cite{Bartal}.
Moreover, as Bartal, Gottlieb, and Krauthgamer show, TSP admits a scheme PTAS with running time $O\big(n^{2^{O(dim)}}\cdot 2^{(2^{dim}/\varepsilon)^{O(dim)}\sqrt{\log n}}\big)$ in any metric space of fixed doubling dimension $dim$~\cite{Bartal}.\medskip

\noindent\textbf{Our contributions.}
Surprisingly, the existence of a polynomial-time approximation scheme for Max~TSP in fixed doubling dimensions was still an open question.
The \emph{doubling dimension} of a metric space is the smallest value $dim\ge 0$ such that every ball in this space can be covered by $2^{dim}$ balls of half the radius.
%A doubling space, i.e., a metric space of bounded doubling dimension, may be considered as a generalization of a fixed-dimensional normed space $(\mathbb R^d,\|.\|)$ (see Remark~\ref{rem2}).
%This generalization seems very useful since, unlike the metrics induced by vector norms, a doubling metric may be not translation invariant and not homogeneous that is rather actual for real life.
A doubling space, i.e., a metric space of bounded doubling dimension, seems to be a natural and useful generalization of a fixed-dimensional normed space since, unlike the metrics induced by vector norms, a doubling metric may be not translation invariant and not homogeneous, which is relevant to real-life distance functions.

We show that, for any $\varepsilon\in(0,1)$, a $(1-\varepsilon)$-approximate solution of the maximum traveling salesman problem in an arbitrary metric space of doubling dimension $dim$ can be computed in time $O\big(2^{(2/\varepsilon)^{2dim+1}}+n^3\big)$.
Thus, in the case of fixed doubling dimension, we have a scheme \mbox{EPTAS}, which is the first polynomial-time approximation scheme for Max~TSP in a doubling space.
Additionally, we propose an $O(n^3)$-time approximation algorithm which computes asymptotically optimal solutions of the problem in fixed and sublogarithmic doubling dimensions, i.e., when $dim=o(\log n)$.
The relative error of this algorithm is estimated as $(11/6)/n^{\frac{1}{2dim+1}}$.

Our technique is pretty simple and is based on combining cycles in cycle covers of the input graph.
The key statement we use is the observation that, in the case of a low doubling dimension, the number of cycles in any cycle cover can be reduced to a small value with a small relative loss of the weight (Lemma~\ref{lem2}).
It allows to get a Hamiltonian cycle whose total weight is close to that of an optimal cycle cover.

\section{Basic definitions and properties}
%Let us define basic concepts we will use and prove basic statements underlying the proposed algorithms.
%
A \emph{metric space} is an arbitrary set ${\cal M}$ with a non-negative distance function $dist$ which is defined for each pair $x,y\in{\cal M}$ and satisfies the triangle inequality and the symmetry axiom.
Given a metric space $({\cal M},dist)$, a \emph{ball} of radius $r$ in this space centered at a point $x\in{\cal M}$ is the set $B(x,r)=\{y\in{\cal M}\,|\,dist(x,y)\le r\}$.
The \emph{doubling dimension} of a metric space is the smallest value $dim\ge 0$ such that every ball in this space can be covered by $2^{dim}$ balls of half the radius.

\begin{rem}\label{rem1}
It is easy to see that, if a metric space is of doubling dimension at most $dim$, then each $r$-radius ball in this space can be covered by $(2/\delta)^{dim}$ balls of radius $\delta r$, where $\delta$ is any value from $(0,1)$.
Indeed, by induction, an $r$-radius ball can be covered by $2^{i\cdot dim}$ balls of radius $r/2^i$, $i=1,2,\dots$.
Hence, by selecting the integer $i$ for which $1/2^i\le\delta<1/2^{i-1}$, we obtain $2^{i\cdot dim}<(2/\delta)^{dim}$ covering balls of radius $r/2^i\le\delta r$.
\end{rem}

%\begin{remark}\label{rem2}
%By simple volume arguments, vector space $\mathbb R^d$ with the metric induced by any vector norm $\|.\|$ is of doubling dimension $\Theta(d)$.
%Indeed, by simple volume arguments, any convex body in $\mathbb R^d$ can be covered by $N$ translates of a $1/2$-scaled copy of this body, where $2^d\le N\le 5^d$.
%So the doubling dimension of $(\mathbb R^d,\|.\|)$ is in $[d,\,d\log_25]$.
%Indeed, any convex body in $\mathbb R^d$ can be covered by $(1+2/\delta)^d$ by $\delta$-scaled copies of this body (e.g., see \cite{VempalaSTOC}).
%Applying it to $\delta=1/2$, we obtain a covering of size~$5^d$.
%So we have $dim\le\log_2(5^d)\approx 2.32\,d$.
%\end{remark}

Suppose that we are given a set $V$ of $n$ points in ${\cal M}$ and also all the pairwise distances $dist(a,b)$, $a,b\in V$.
Denote by $G[V]$ the complete weighted undirected graph on the vertex set $V$ in which the weight of every edge $\{a,b\}$ is defined as $dist(a,b)$.
The metric Max~TSP asks to find a maximum-weight Hamiltonian cycle in $G[V]$.

In short, the suggested algorithm can be described as follows.
We start with constructing a maximum-weight cycle cover of the graph $G[V]$, i.e., a maximum-weight spanning subgraph of this graph in which every connected component is a cycle.
Then, based on properties of doubling metrics, we significantly reduce the number of cycles in this cycle cover with a small weight loss.
Finally, the remaining cycles are combined into one, also with a slight weight loss, by using the known method from~\cite{Serdyukov1991}.

To reduce the number of cycles in the cycle cover, we repeatedly patch two of the cycles into one.
Below, we define the main notion we use in this procedure:

\begin{definition*}
Let $c_1$, $c_2$ be vertex-disjoint cycles in $G[V]$ and $\{a_i,b_i\}$ be any edge in $c_i$, $i=1,2$.
A \emph{$\delta$-gluing of the cycles $c_1,c_2$ on the edges $\{a_1,b_1\}$, $\{a_2,b_2\}$}, where $\delta\in(0,1)$, is a combining of these cycles into one by replacing the pair of edges $\{a_1,b_1\}$, $\{a_2,b_2\}$ by one of two pairs $\{a_1,b_2\}$, $\{a_2,b_1\}$ or\, $\{a_1,a_2\}$, $\{b_1,b_2\}$ such that the total weight of the replacing pair is at least $(1-\delta)\big(dist(a_1,b_1)+dist(a_2,b_2)\big)$.
\end{definition*}

A $\delta$-gluing allows to patch the cycles $c_1,c_2$ into one so that we lose at most $\delta$ of the total weight of the edges $\{a_1,b_1\}$, $\{a_2,b_2\}$.
%Note that, by the symmetry axiom, the order of vertices in each of these edges does not matter: If the cycles $c_1,c_2$ admit a $\delta$-gluing on the edges $(a_1,b_1)$, $(a_2,b_2)$, then they admit a $\delta$-gluing on the edges $(b_1,a_1),(a_2,b_2)$, on the edges $(a_1,b_1),(b_2,a_2)$, and on the edges $(b_1,a_1),(b_2,a_2)$.

Let us consider arbitrary vertex-disjoint cycles $c_1,\dots,c_k$ forming a cycle cover of the graph $G[V]$, select any edge $\{a_i,b_i\}$ in each cycle $c_i$, $i=1,\dots,k$, and suppose that $\{a_\tau,b_\tau\}$ is a shortest edge among all $\{a_i,b_i\}$.
Define the value
$$R_\tau(\{a_i,b_i\}_{i=1}^k)=\max_{v\in S}dist(\{a_\tau,b_\tau\},v),$$
where $S=\{a_1,b_1,\dots,a_k,b_k\}$ and $dist(\{a,b\},v)=\min\{dist(a,v),dist(b,v)\}$.
The key statements underlying our algorithm are the following lemmas.

\begin{lemma}\label{lem1} %Lemma~1
If no pair of cycles $c_p,c_q$ admits a $\delta$-gluing on the edges $\{a_p,b_p\}$, $\{a_q,b_q\}$, $p,q\in\{1,\dots,k\}$, then $R_\tau(\{a_i,b_i\}_{i=1}^k)<t/\delta-t$, where $t=dist(a_\tau,b_\tau)$.
\end{lemma}

\begin{proof}
By the construction, the value of $R=R_\tau(\{a_i,b_i\}_{i=1}^k)$ equals the distance between some vertex $u\in\{a_\tau,b_\tau\}$ and some vertex $v\in\{a_\ell,b_\ell\}$, where $\ell\in\{1,\dots,k\}$.
Let $u'$ and $v'$ be the other endpoints of the edges $\{a_\tau,b_\tau\}$ and $\{a_\ell,b_\ell\}$, i.e., those for which $\{a_\tau,b_\tau\}=\{u,u'\}$ and $\{a_\ell,b_\ell\}=\{v,v'\}$.
Then, by the triangle inequality, we have $dist(u,v)+dist(u,v')\ge dist(v,v')$ (see Fig.~\ref{fig:1}).
At the same time, the choice of $u$ implies that $dist(u,v)=dist(\{u,u'\},v)\le dist(u',v)$.
So
$$dist(u',v)+dist(u,v')\ge dist(u,v)+dist(u,v')\ge\max\{dist(v,v'),dist(u,v)\}.$$

\begin{figure}[!ht]
\centering
\captionsetup{justification=centering}
\includegraphics[scale=1]{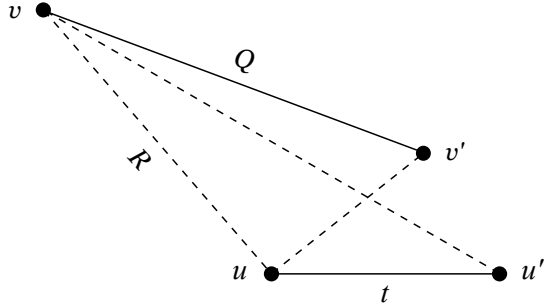}
\caption{Proof of Lemma~\ref{lem1}}
\label{fig:1}
\end{figure}

By the condition, the cycles $c_\tau$ and $c_\ell$ do not admit a $\delta$-gluing on the edges $\{a_\tau,b_\tau\}$, $\{a_\ell,b_\ell\}$.
Then, by the above, we have $\max\{Q,R\}<(Q+t)(1-\delta)$, where $Q=dist(v,v')$.
But the maximum of two values is at least any of their convex combinations, so $\max\{Q,R\}\ge Q(1-\delta)+R\delta$.
It follows that $Q(1-\delta)+R\delta<(Q+t)(1-\delta)$, which implies the inequality $R<t/\delta-t$.
The lemma is proved.
\hfill$\Box$
%\qed
\end{proof}

\begin{lemma}\label{lem2} %Lemma~2
If the space $({\cal M},dist)$ is of doubling dimension at most $dim$ and no pair of cycles $c_p,c_q$ admits a $\delta$-gluing on the edges $\{a_p,b_p\}$, $\{a_q,b_q\}$, $p,q\in\{1,\dots,k\}$, then $k$ is at most $(2/\delta)^{2dim}/2$.
\end{lemma}

\begin{proof}
By Lemma~\ref{lem1} and the triangle inequality, the set $S=\{a_1,b_1,\dots,a_k,b_k\}$ is contained in the ball $B(u,t/\delta)$, where $u\in\{a_\tau,b_\tau\}$.
On the other hand, the definition of doubling dimension implies that this ball can be covered by $(4/\delta^2)^{dim}=(2/\delta)^{2dim}$ balls of radius $t\delta/2$ (see Remark~\ref{rem1}).

Suppose that $k>(2/\delta)^{2dim}/2$.
Then $|S|=2k>(2/\delta)^{2dim}$, so there exists a pair of vertices $x,y\in S$, $x\ne y$, lying in one of the $(t\delta/2)$-radius balls which cover the ball $B(u,t/\delta)$.
But $t$ is the weight of a shortest edge among all $\{a_i,b_i\}$.
Hence, by the triangle inequality and since $\delta<1$, we have
$$dist(x,y)\le t\delta<t\le dist(a_i,b_i)$$
for all $i\in\{1,\dots,k\}$.
It follows that $x$ and $y$ are vertices of different edges $\{a_i,b_i\}$, say, $x\in\{a_p,b_p\}$, $y\in\{a_q,b_q\}$ for some $p,q\in\{1,\dots,k\}$, $p\ne q$.
Denote by $x'$ and $y'$ the other endpoints of $\{a_p,b_p\}$, $\{a_q,b_q\}$, i.e., $\{x,x'\}=\{a_p,b_p\}$ and $\{y,y'\}=\{a_q,b_q\}$.
Then, by the above and by the axioms of metric, we obtain the inequalities
$$dist(y,x')\ge dist(x,x')-t\delta\ge dist(x,x')(1-\delta),$$
$$dist(x,y')\ge dist(y,y')-t\delta\ge dist(y,y')(1-\delta),$$
which imply that the cycles $c_p,c_q$ admit a $\delta$-gluing on the edges $\{a_p,b_p\}$, $\{a_q,b_q\}$.
The lemma is proved.
\hfill$\Box$
%\qed
\end{proof}

Next, we recall the following result of Serdyukov, which will be useful for us:

\begin{fact}{\rm\cite{Serdyukov1991,BGS}\,}\label{fact1} %Fact~1
Let $G$ be an $n$-vertex complete weighted graph where the distances between vertices satisfy the triangle inequality and let $C$ be a cycle cover of $G$ which consists of $k$ cycles.
Then $C$ can be combined into a Hamiltonian cycle of total weight at least $(1-1/n)^{k-1}$ of that of $C$ by an $O(kn)$-time algorithm.
\end{fact}

\section{Algorithms}
Lemma~\ref{lem2} and Fact~\ref{fact1} prompt an idea how to patch the cycles of an optimal cycle cover into one with a small weight loss: While it is possible, we will perform $\delta$-gluings of these cycles and then combine the remaining, rather small, number of cycles by using the Serdyukov's result.
It can be formalized as follows:\medskip

\noindent\textbf{Algorithm~${\cal A}$.}\smallskip

\noindent\emph{Input}: a set $V$ of $n$ points in ${\cal M}$; the distances $dist(a,b)$ for all $a,b\in V$; a parameter $\delta\in(0,1)$.
\emph{Output}: a Hamiltonian cycle $H$ in the graph $G[V]$.\medskip

\noindent\emph{Step~{\rm 1}}:
By using the $O(n^3)$-time algorithm from \cite{Gabow}, find a maximum-weight cycle cover $C_0$ of the graph $G[V]$;
construct a set $E_0$ which includes exactly two minimum-weight edges of each cycle in $C_0$.\medskip

\noindent\emph{Step~{\rm 2}}:
Set $C=C_0$ and, while it is possible, repeat the following operations.
Denote by $c_1,\dots,c_k$ the cycles in $C$ and, for each $i=1,\dots,k$, select any edge $\{a_i,b_i\}\in E_0$ in the cycle~$c_i$.
Find any pair of cycles $c_p,c_q$, $p,q\in\{1,\dots,k\}$, which admit a $\delta$-gluing on the edges $\{a_p,b_p\}$, $\{a_q,b_q\}$ and update $C$ by performing this $\delta$-gluing.\medskip

\noindent\emph{Step~{\rm 3}}:
Apply the algorithm from Fact~\ref{fact1} to the cycle cover $C$ and return the resulting Hamiltonian cycle $H$.

\begin{theorem}\label{th1}
If the space $({\cal M},dist)$ is of doubling dimension at most $dim$, then Algorithm~${\cal A}$ finds a Hamiltonian cycle of total weight at least $1-(2/3)\delta-(2/\delta)^{2dim}/(2n)$ of that of an optimal cycle cover in time $O(n^3)$.
\end{theorem}

\begin{proof}
It is easy to prove by induction that, at each iteration of Step~2, every cycle in the cycle cover $C$ contains at least two edges from the set $E_0$.
So we always may select a required edge $\{a_i,b_i\}$ in every cycle $c_i$ in $C$.
Then, according to Lemma~\ref{lem2}, the number of cycles in $C$ is reduced to $k\le(2/\delta)^{2dim}/2$ by the end of Step~2.

Next, at each $\delta$-gluing, we replace some edges $e_1,e_2\in E_0$ in the cycle cover $C$ by some edges $e'_1,e'_2$ which connect the endpoints of $e_1$ with those of $e_2$.
Since the edges $e_1,e_2$ belong to different cycles in $C_0$, then the edges $e'_1,e'_2$ can not belong to $C_0$, so $e'_1,e'_2\notin E_0$.
Therefore, after any edge is removed from the cycle cover $C$, it is no longer included in $C$ on further $\delta$-gluings.
At the same time, the total weight of $e'_1,e'_2$ is at least $1-\delta$ of that of $e_1,e_2$.
It follows that, during Step~2, the total weight of $C$ is decreased at most by $\delta$ of the total weight of $E_0$.
But, by the construction of the set $E_0$, its total weight is at most $2/3$ of that of $C_0$.
Hence, by the end of Step~2, the total weight of $C$ is at least $1-(2/3)\delta$ of that of~$C_0$.

Finally, by Fact~\ref{fact1}, the Hamiltonian cycle $H$ we obtain at Step~3 is of total weight at least $(1-1/n)^{k-1}>1-k/n$ of that of~$C$.
So the total weight of $H$ is at least
$$\big(1-(2/3)\delta\big)(1-k/n)>1-(2/3)\delta-(2/\delta)^{2dim}/(2n)$$
of that of $C_0$.

It remains to estimate the time complexity of Algorithm~${\cal A}$.
At Step~1, we construct an optimal cycle cover by using the algorithm from \cite{Gabow} in time~$O(n^3)$.
Each iteration of Step~2 can be performed in time $O(n^2)$, while the number of these iterations is~$O(n)$.
Step~3 takes time $O(kn)=O(n^2)$ by Fact~\ref{fact1}.
Thus, the running time of Algorithm~${\cal A}$ is~$O(n^3)$.
The theorem is proved.
\hfill$\Box$
%\qed
\end{proof}

As a corollary, we obtain an efficient polynomial-time approximation scheme for Max~TSP in a metric space of fixed doubling dimension:

\begin{theorem}\label{th2}
Max~TSP in a metric space of doubling dimension at most $dim$ admits an approximation scheme which, for any $\varepsilon\in(0,1)$, finds a $(1-\varepsilon)$-approximate solution of the problem in time $O\big(2^{(2/\varepsilon)^{2dim+1}}+n^3\big)$.
\end{theorem}

\begin{proof}
If $\varepsilon\ge 1/6$, then we get a $(1-\varepsilon)$-approximate solution of Max~TSP by using the $O(n^3)$-time $5/6$-approximation algorithm of Kostochka and Serdyukov~\cite{Serdyukov1985,BGS}.
Suppose that $\varepsilon<1/6$.
In this case, we set $\delta=(12/11)\varepsilon$ and, since $\delta<1$, obtain that the approximation ratio of Algorithm~${\cal A}$ is at least $1-(8/11)\varepsilon-((11/6)/\varepsilon)^{2dim}/(2n)$.

If $n$ is greater than $n(\varepsilon)=((11/6)/\varepsilon)^{2dim+1}$, then the term $((11/6)/\varepsilon)^{2dim}/(2n)$ is less than $(3/11)\varepsilon$, so Algorithm~${\cal A}$ outputs a Hamiltonian cycle with ap\-p\-ro\-xi\-ma\-ti\-on factor at least $1-\varepsilon$.
If $n\le n(\varepsilon)$, then we compute an optimal solution of Max~TSP by using the exact $O(2^nn^2)$-time dynamic-programming algorithm for the usual TSP from \cite{Bellman,Held}.
To reduce Max~TSP to TSP, we replace the weight $w(e)$ of every edge $e$ by the value $w-w(e)$, where $w$ is the maximum edge weight.

Thus, in any case, we get a $(1-\varepsilon)$-approximate solution of Max~TSP in time
$$O\big(\max\big\{2^{((11/6)/\varepsilon)^{2dim+1}}((11/6)/\varepsilon)^{4dim+2},\ n^3\big\}\big)=O\big(2^{(2/\varepsilon)^{2dim+1}}+n^3\big).$$
The theorem is proved.
\hfill$\Box$
%\qed
\end{proof}

Another corollary of Theorem~\ref{th1} is a simple cubic-time approximation algorithm which computes asymptotically optimal solutions of Max~TSP in fixed and ``slowly growing'' doubling dimensions:

\begin{theorem}\label{th3}
Max~TSP in a metric space of doubling dimension at most $dim$ admits an $O(n^3)$-time approximation algorithm with relative error at most $(11/6)/n^{\frac{1}{2dim+1}}$.
\end{theorem}

\begin{proof}
By Theorem~\ref{th1}, the relative error of Algorithm~${\cal A}$ is bounded by the value of $err(\delta)=(2/3)\delta+(2/\delta)^q/(2n)$, where $q=2dim$.
If $n>2^{q+1}$, then we apply this algorithm with the parameter $\delta=2/n^{\frac{1}{q+1}}$.
In this case, we have $\delta<1$ and
$$err(\delta)=(4/3)/n^{\frac{1}{q+1}}+n^{\frac{q}{q+1}}/(2n)=(4/3+1/2)/n^{\frac{1}{q+1}}=(11/6)/n^{\frac{1}{q+1}}.$$

If $n\le 2^{q+1}$, then we use the $O(n^3)$-time $5/6$-approximation algorithm of Kostochka and Serdyukov \cite{Serdyukov1985,BGS}.
In this case, we also obtain a solution with relative error at most $(11/6)/n^{\frac{1}{q+1}}$ since $(11/6)/(2^{q+1})^{\frac{1}{q+1}}=11/12>1/6$.
%Let us minimize the function $err(\delta)$ over all $\delta\in(0,1)$.
%The derivative of this function is $2/3-q2^{q-1}/(n\delta^{q+1})$.
%So the minimum of $err(\delta)$ is attained at the value of $\delta^*=(1.5q2^{q-1}/n)^{\frac{1}{q+1}}$, which belongs to the interval $(0,1)$ if $n>1.5q2^{q-1}$.
%In this case, the value of $err(\delta^*)$ equals
%\begin{eqnarray*}
%\frac{2}{3}\Big(\frac{1.5q2^{q-1}}{n}\Big)^{\frac{1}{q+1}}+\frac{2^{q-1}}{n}\Big(\frac{n}{1.5q2^{q-1}}\Big)^{\frac{q}{q+1}}=
%\Big(\frac{2^{q-1}}{n}\Big)^{\frac{1}{q+1}}\Big(\frac{2}{3}(1.5q)^{\frac{1}{q+1}}+(1.5q)^{-\frac{q}{q+1}}\Big).
%\end{eqnarray*}
%But $\displaystyle 2^{\frac{q-1}{q+1}}\Big(\frac{2}{3}(1.5q)^{\frac{1}{q+1}}+(1.5q)^{-\frac{q}{q+1}}\Big)<1.84$ for all $q>0$.
%Thus, if $n>1.5q2^{q-1}$, then the relative error of Algorithm~${\cal A}$ with the optimized value of the parameter $\delta$ is less than $1.84/n^{\frac{1}{q+1}}$.
%
%If $n\le 1.5q2^{q-1}$, then we will use the $O(n^3)$-time $5/6$-approximation algorithm of Kostochka and Serdyukov \cite{Serdyukov1985,BGS}.
%In this case, we also obtain solutions with relative error at most $1.84/n^{\frac{1}{q+1}}$ since $1.84/(1.5q2^{q-1})^{\frac{1}{q+1}}>1/6$ for all $q>0$.
The theorem is proved.
\hfill$\Box$
%\qed
\end{proof}

In the case when $dim=o(\ln n)$, the relative error of the algorithm described in Theorem~\ref{th3} is at most $(11/6)/n^{\frac{1}{o(\ln n)}}=(11/6)/e^{\frac{\ln n}{o(\ln n)}}\rightarrow 0$ as $n\rightarrow\infty$.
So we have a polynomial-time asymptotically exact algorithm for the metric Max~TSP in fixed and sublogarithmic doubling dimensions.

%\begin{remark}\label{rem3}
%The Euclidean Max~TSP in space of dimension $d$ admits an $O(n^3)$-time approximation algorithm with relative error $O(1/n^{\frac{2}{d+1}})$ \cite{Serdyukov1987}.
%A generalized version of this algorithm allows to compute an approximate solution of Max~TSP in any fixed-dimensional normed space $(\mathbb R^d,\|.\|)$ with relative error $O(1/n^{\frac{1}{d+1}})$ \cite{Shen2010,Shen2014}.
%It follows that, in the case of fixed $d$, the geometric Max~TSP admits an approximation scheme EPTAS with running time $2^{O(1/\varepsilon)^{d+1}}+O(n^3)$.
%\end{remark}

\section{Conclusion}
We propose an efficient polynomial-time approximation scheme (EPTAS) for the ma\-xi\-mum traveling salesman problem in a metric space of fixed doubling dimension.
Additionally, we describe a cubic-time asymptotically exact algorithm for this problem in fixed and sublogarithmic doubling dimensions.

A natural direction for future work is constructing an approximation scheme which is ``efficient'' not only in the sense of the definition of an EPTAS but also in the practical sense, i.e., which finds close-to-optimal solutions of the problem in reasonable time.
However, this may not be easy since, unless P$=$NP, even the $3$-dimensional Euclidean Max~TSP does not admit a fully polynomial-time approximation scheme.
%A much simpler area of possible research is extending the obtained results to a number of related maximization routing problems such as the maximum $m$-peripatetic salesman problem and the maximum $k$-cycle cover.

An open question is the existence of an approximation scheme for sublogarithmic doubling dimensions.
This case is interesting in that any $n$-point metric space is of doubling dimension at most $\log_2n$, while Max~TSP is APX-hard in the general metric setting.
In this regard, the dimension $dim=o(\ln n)$ may be a boundary case for instances of Max~TSP which admit a near-optimal approximation.

%\vspace{2em}\noindent\textbf{Acknowledgments.}
%The study was carried out within the framework of the state contract of the Sobolev Institute of Mathematics (project 0314-2019-0014).

% Non-BibTeX users please use

\end{document}